\documentclass[10pt]{article}
\pdfoutput=1
\usepackage{hyperref}
\usepackage{color}
\newif\ifFOCS \FOCSfalse

\usepackage{verbatim}
\usepackage[cmex10]{amsmath}
\usepackage{amssymb}
\usepackage{amsthm}
\newcommand{\opt}{\mathsf{opt}}
\renewcommand{\c}{\mathbf{c}}

\newcommand{\E}{\mathcal{E}}
\newcommand{\Exp}{\mathbb{E}}

\newcommand{\Z}{\mathbb{Z}}

\newcommand{\Lin}{\mathrm{Lin}}

\newcommand{\FF}{\mathfrak{F}}

\newcommand{\C}{\mathbf{C}}
\newcommand{\A}{\mathbf{A}}

\newcommand{\B}{\mathbf{B}}
\renewcommand{\b}{\mathbf{b}}
\renewcommand{\O}{\mathcal{O}}

\renewcommand{\l}{\mathbf{l}}

\newcommand{\D}{\mathcal{D}}
\newcommand{\up}{\mathsf{up}}
\newcommand{\dn}{\mathsf{dn}}
\newcommand{\G}{\mathcal{G}}

\renewcommand{\P}{\mathbf{P}}
\newcommand{\R}{\mathbb{R}}

\renewcommand{\L}{\mathcal{L}}
\renewcommand{\d}{\mathbf{d}}
\newcommand{\one}{\mathbf{1}}
\newcommand{\simp}{\triangle}

\newcommand{\X}{\mathcal{X}}
\newcommand{\Y}{\mathcal{Y}}
\newcommand{\OT}{\Theta}

\newcommand{\V}{\mathcal{V}}

\newcommand{\x}{\mathbf{x}}

\ifFOCS
\newcommand{\FOCSspace}[1]{}

\else
\newcommand{\FOCSspace}[1]{}

\fi
\newcommand{\grant}{Research supported by NSF Grant CCF-1410022}
\newtheorem{theorem}{Theorem}[section]
\newtheorem{corollary}[theorem]{Corollary}

\newtheorem{defn}[theorem]{Definition}
\newtheorem{lemma}[theorem]{Lemma}

\begin{document}
\ifFOCS
\else
\title{Generalized Preconditioning and Network Flow Problems}
\author{Jonah Sherman\thanks{\grant}\\University of California, Berkeley}
\date{\today (preliminary draft)}
\fi
\maketitle
\ifFOCS
\else
\thispagestyle{empty}
\fi

\begin{abstract}
We consider approximation algorithms for the problem of finding $x$ of minimal norm $\|x\|$ satisfying a linear system $\A x = \b$, where the norm $\|\cdot \|$ is arbitrary and generally non-Euclidean.  We show a simple general technique for composing solvers, converting iterative solvers with residual error $\|\A x - \b\| \leq  t^{-\Omega(1)}$ into solvers with residual error $\exp(-\Omega(t))$, at the cost of an increase in $\|x\|$, by recursively invoking the solver on the residual problem $\tilde{\b} = \b - \A x$.  Convergence of the composed solvers depends strongly on a generalization of the classical condition number to general norms, reducing the task of designing algorithms for many such problems to that of designing a \emph{generalized preconditioner} for $\A$.  The new ideas significantly generalize those introduced by the author's earlier work on maximum flow, making them more widely applicable.  

As an application of the new technique, we present a nearly-linear time approximation algorithm for uncapacitated minimum-cost flow on undirected graphs.  Given an undirected graph with $m$ edges labelled with costs, and $n$ vertices labelled with demands, the algorithm takes $\epsilon^{-2}m^{1+o(1)}$-time and outputs a flow routing the demands with total cost at most $(1+\epsilon)$ times larger than minimal, along with a dual solution proving near-optimality.  The generalized preconditioner is obtained by embedding the cost metric into $\ell_1$, and then considering a simple hierarchical routing scheme in $\ell_1$ where demands initially supported on a dense lattice are pulled from a sparser lattice by randomly rounding unaligned coordinates to their aligned neighbors.  Analysis of the generalized condition number for the corresponding preconditioner follows that of the classical multigrid algorithm for lattice Laplacian systems.
\end{abstract}

\section{Introduction}
A fundamental problem in optimization theory is that of finding solutions $x$ of minimum-norm $\|x\|$ to rectangular linear systems $\A x = \b$.  Such solvers have extensive applications, due to the fact that many practical optimization problems reduce to minimum norm problems.  When the norm is Euclidean, classical iterative solvers such as steepest descent and conjugate gradient methods produce approximately optimal solutions with residual error $\|\A x - \b\|$ exponentially small in iteration count,
The rate of exponential convergence depends strongly on the \emph{condition number} of the square matrix  $\A \A^*$.

Therefore, the algorithm design process for such problems typically consists entirely of efficiently constructing \emph{preconditioners}: easily computable left-cancellable operators $\P$, for which the transformed problem $\P\A = \P \b$ is well-conditioned so iterative methods converge rapidly.  In a seminal work, Spielman and Teng\cite{ST} present a nearly-linear-time algorithm to construct preconditioners with condition number polylogarithmic in problem dimension for a large class of operators $\A$, including Laplacian systems and others arising from discretization of elliptic PDEs.

There are many fundamental optimization problems that can be expressed as minimum-norm problems with respect to more general, non-Euclidean norms, including some statistical inference problems and network flow problems on graphs.  In particular, the fundamental \emph{maximum-flow} and \emph{uncapacitated minimum-cost flow} problems reduce respectively to $\ell_\infty$ and $\ell_1$ minimum-norm problems in the undirected case.  

For such problems,  there is presently no known black-box iterative solver analagous to those for $\ell_2$ that converge exponentially with rate independent of problem dimension.  Therefore, one is forced to choose between interior point methods to obtain zero or exponentially small error\cite{Boyd}, at the cost of iteration count depending strongly on problem dimension, or alternative methods that have error only polynomially small with respect to iterations\cite{PST,FS,NesterovSmooth}.

An important idea implicitly used by classical iterative $\ell_2$ solvers is \emph{residual recursion}: reducing the error of an approximate solution $x$ by recursively applying the solver to $\tilde{\b} = \b - \A x$.  Such recursion is implicit in those solvers due to the fact that in $\ell_2$, the map $\b \mapsto \x_{\opt}$ is linear, with $\x_{\opt} = (\A\A^*)^{+} \b$.  Therefore, the classical methods are effectively recursing on \emph{every} iteration, with no distinction between iterations and recursive solves.  In general norms, an $\x_{\opt}$ may not be unique, and there is generally no linear operator mapping $\b$ to some $\x_{\opt}$.

In this paper, we make two primary contibutions.  Our first contribution is to extend analysis of residual recursion to arbitrary norms.  We present a general and modular framework for efficiently solving minimum-norm problems by composing simple well-known base solvers.  The main tool is the \emph{composition lemma}, describing the approximation parameters of a solver composed via residual recursion of two black-box solvers.  The resulting parameters depend on a generalization of the condition number to arbitrary norms; therefore, much like $\ell_2$, the algorithm design process is reduced to constructing good \emph{generalized preconditioners}.  Our second contribution is to present, as a practical application of those tools, a nearly-linear time approximation algorithm for the uncapacitated minimum-cost flow problem on undirected graphs.  Having established the former framework, the latter algorithm is entirely specified by the construction of a generalized preconditioner, and analyzed by bounding its generalized condition number.  

Our framework builds upon earlier work\cite{Sherman13}, where we leveraged ideas introduced by Spielman and Teng to obtain nearly-linear time approximation algorithms for undirected maximum flow.  We have since realized that some of the techniques applied to extend $\ell_2$-flows to $\ell_\infty$-flows are in no way specific to flows, 
and indeed the composition framework presented here is a generalization of those ideas to non-Euclidean minimum-norm problems.  Moreover, a significantly simpler max-flow algorithm may be recovered using black-box solvers and general composition as presented here, requiring only the truly flow-specific part of the earlier work: the \emph{congestion approximator}, now understood to be a specific instance of a more widely-applicable \emph{generalized preconditioner}.

We proceed to discuss those contributions separately in more detail, and outline the corresponding sections of the paper.  We also discuss some related work.

\subsection{Minimum Norm Problems}
In section \ref{approx}, we define the minimum norm problem more precisely, and introduce a useful bicriteria notion of approximation we call $(\alpha,\beta)$-solutions.  In that notion, $\alpha$ quantifies how much $\|x\|$ relatively exceeds $\|x_\opt\|$, and $\beta$ quantifies the residual error $\|\A x - \b\|$, with an appropriate relative scale factor.  We may succinctly describe some frequently used algorithms as $(\alpha,\beta)$-solvers; we discuss some examples in section \ref{lpsolve}, including steepest descent and conjugate gradient for $\ell_2$, and \emph{multiplicative weights} for $\ell_1$ and $\ell_\infty$.

In section \ref{reccomp}, we present the \emph{composition lemma}, which shows that by composing a $(\alpha_1,\beta_1)$-solver with a $(\alpha_2,\beta_2)$-solver in a black-box manner, we obtain a $(\alpha_3,\beta_3)$-solver with different parameter tradeoffs. 
The composed algorithm's parameters $(\alpha_3,\beta_3)$ depends crucially on a generalization of the classical condition number to rectangular matrices in general norms; Demko\cite{Demko} provides essentially the exact such generalization needed, so we briefly recall that definition before stating and proving the composition lemma.  
The proof of the composition lemma is not difficult; it follows quite easily once the right definitions of $(\alpha,\beta)$-solutions and a generalized condition number have been established.  Nevertheless, the modularity of the composition lemma proves to quite useful.  
By recursively composing a solver with itself $t$ times, we may decrease $\beta$ exponentially with $t$, at the cost of a single \emph{fixed} increase in $\alpha$. By composing solvers with different parameters, better parameters may be obtained than any single solver alone.  A particularly useful example uses a $(M,0)$-solver, where $M$ is uselessly-large on its own, to terminate a chain of $t$ solvers and obtain zero error.  As the final step in the chain, we get all of the benefits (zero error), and almost none of the costs, as it contributes  $M 2^{-t}$ instead of $M$ to the final solver.

By composing the existing solvers discussed in section \ref{lpsolve}, we obtain solvers for $\ell_1$ and $\ell_\infty$ problems with residual error exponentially small in $t$.  We state the final composed algorithms parameters.  As a result, the algorithm design problem for such problems is reduced to that of designing a \emph{generalized preconditioner} for $\A$.  We briefly discuss such preconditioning in section \ref{precondsec}.
\subsection{Minimum Cost Flow}
The second part of this paper applies generalized preconditioning to solve a fundamental problem in network optimization.
In the \emph{uncapacitated minimum-cost flow problem} on undirected graphs, we are given a connected graph $G$ with $m$ edges, each annotated with a cost, and a specified demand at each vertex.  The problem is to find an edge flow with vertex divergences equal to the specified demands, of minimal total cost.  Our main result is a randomized algorithm that, given such a problem, takes $\epsilon^{-2}m^{1+o(1)}$ time and outputs a flow meeting the demands with cost at most $(1+\epsilon)$-times that of optimal.  We define various flow-related terms and give a more precise statement of the problem and our result in section \ref{graphsec}.  

We describe and analyze the algorithm in section \ref{mincostsec}.  Having built up our general tools in earlier sections, our task is reduced to the design and analysis of a generalized preconditioner for min-cost flow.  We begin by interpreting the edge costs as lengths inducing a metric on the graph, and then apply Bourgain's small-distortion embedding\cite{Bourgain85} into $\ell_1$.  That is, by paying a small distortion factor in our final condition number, we may focus entirely on designing a preconditioner for min-cost flow in $\ell_1$ space.  The construction of our precondtitioner is based on an extremely simple hierarchical routing scheme, inspired by a combination of the multigrid algorithm for grid Laplacians, and the Barnes-Hut algorithm for $n$-body simulation\cite{BH}.  The routing scheme proceeds on a sequence of increasingly dense lattices $V_0,V_1,\ldots$ in $\ell_1$, where $V_t$ are the lattice points with spacing $2^{-t}$, and the input demands are specified on the densest level $V_T$.  
Starting at level $t = T$, the routing scheme sequentially eliminates the demand supported on $V_t \setminus V_{t-1}$ by pulling the demand from level $t-1$ using a random-rounding based routing:
each vertex $x \in V_t$ picks a random nearest neighbor in $V_{t-1}$, corresponding to randomly rounding its unaligned coordinates, and then pulls a fractional part of its demand from that neighbor via any shortest path.  Afterwards, all demands on $V_t \setminus V_{t-1}$ are met, leaving a reduced problem with demands supported on $V_{t-1}$.  The scheme then recurses on the reduced problem.
At level $0$, the demands are supported on the corners of a hypercube, for which the simple routing scheme of pulling all demand from a uniformly random corner performs well-enough.  Having described the routing scheme, we need not actually carry it out.  Instead, our preconditioner crudely estimates the cost of that routing.
Furthermore, while the routing has been described on an infinite lattice, we'll observe that if the demands are only initially supported on a small set of $n$ vertices, the support remains small throughout.

\subsection{Related Work}
Spielman and Teng present a nearly-linear time algorithm for preconditioning and solving symmetric diagonally dominant matrices, including those of graph Laplacians\cite{ST}.  The ideas introduced in that work have led to breakthroughs in various network cut and flow algorithms for the case of undirected graphs. Christiano \emph{et. al.} apply the solver as a black-box to approximately solve the maximum flow problem in $\OT(m^{4/3})$ time.  Madry\cite{Madry10} opens the box and uses the ideas directly to give a family of algorithms for approximating cut-problems, including a $m^{o(1)}$-approximation in $m^{1+o(1)}$ time.  Both the present author\cite{Sherman13} and Kelner \emph{et. al.}\cite{KLOS} combine madry's ideas with $\ell_\infty$ optimization methods to obtain $(1-\epsilon)$-approxmations to maximum flow in $m^{1+o(1)}\epsilon^{-\O(1)}$ time.  While those two algorithms use similar ideas, the actual path followed to achieve the result is rather different.  Kelner \emph{et. al.} construct an \emph{oblivious routing scheme} that is $n^{o(1)}$-competetive, and then show how to use it to obtain a flow.  The algorithm maintains a demand-respecting flow at all times, and minimizes a potential function measuring edge congestion.  In contrast, our earlier algorithm short-cuts the need to explicitly construct the routing scheme by using Madry's construction directly, maintaining a flow that is neither demand nor capacity respecting, aiming to minimize a certain potential function that measures both the congestion and the demand error.  

For min-cost flow, earlier work considers the more general directed, capacitated case, with integer capacities in $\{1,\ldots,U\}$.   The $\OT(nm \log \log U)$-time double-scaling algorithm of Ahuja, Goldberg, Orlin, and Tarjan\cite{AGOT} remained the fastest algorithm for solving min-cost flow for 25 years between its publication and the Laplacian breakthrough.  Shortly after that breakthrough,  Daitch and Spielman\cite{DaitchSpielman} showed how to use the Laplacian solver with interior-point methods to obtain an $\OT(m^{3/2}\log^2 U)$-time algorithm.  Lee and Sidford further reduce this to $\OT(m\sqrt{n}\log^{\O(1)}U)$ by a general improvement in interior point methods\cite{LeeSidford}.
We are not aware of prior work 
\section{Approximate Solutions}\label{approx}
Let $\X,\Y$ be finite dimensional vector spaces, where $\X$ is also a Banach space, and let $\A \in \Lin(\X,\Y)$ be fixed throughout this section.  We consider the problem of finding a minimal norm pre-image of a specified $\b$ in the image of $\A$; that is, finding $x \in \X$ with $\A x = \b$ and $\|x\|_\X$ minimal.


Let $\x_{\opt}$ be an exact solution, with $\A \x_{\opt} = \b$, and $\|\x_{\opt}\|_\X$ minimal.  There are multiple notions of approximate solutions for this problem.  The most immediate is $x \in \X$ with $\A x = \b$ and 
\begin{equation} \frac{\|x\|}{\|\x_{\opt}\|} \leq \alpha \label{alphaonly} \end{equation}
.  We call such $x$ an $(\alpha,0)$-solution; we shall be interested in finding $(1+\epsilon, 0)$-solutions for small $\epsilon$.
A weaker notion of approximation is obtained by further relaxing $\A x = \b$ to $\A x \approxeq \b$.  Quantifying that requires more structure on $\Y$, so let us further assume $\Y$ to also be a Banach space.  In that case, we say $x$ is an $(\alpha,\beta)$-solution if equation \ref{alphaonly} holds and
\begin{equation}
\frac{\|\A x - \b\|}{\|\A\|\|\x_{\opt}\|} \leq \beta  \label{alphabeta} 
 \end{equation} 

The practical utility of such solutions depends on the application.  However, the weaker notion has the distinct advantage of being approachable by a larger family of algorithms, such as penalty and dual methods, by avoiding equality constraints.  We discuss such existing algorithms in section \ref{lpsolve}, but mention that they typically yield $(1,\epsilon)$-solutions after some number of iterations.  For $\ell_2$, the iteration dependency on $\epsilon$ is $\O(\log (1/\epsilon))$, while for some more general norms it is $\epsilon^{-\O(1)}$.

\section{Recursive Composition and Generalized Condition Numbers}\label{reccomp}
We now consider how to trade an increase in $\alpha$ for a decrease in $\beta$.  Residual recursion suggests a natural strategy: after finding an $(\alpha,\beta)$-solution, recurse on $\tilde{\b} = \b - \A x$.  More precisely, we define the \emph{composition} of two algorithms as follows.

\begin{defn}
Let $F_i$ be an $(\alpha_i,\beta_i)$-algorithm for $\A$, for $i \in \{1,2\}$.  The \emph{composition} $F_2 \circ F_1$ takes input $\b$, and first runs $F_1$ on $\b$ to obtain $x$.  Next, setting $\tilde{\b} = \b - \A x$, it runs $F_2$ on input $\tilde{\b}$ to obtain $\tilde{x}$.  Finally, it outputs $x + \tilde{x}$.
\end{defn}

Success of composition depends on $\|\tilde{\b}\|_\Y$ being small implying $\tilde{\b}$ has a small-norm pre-image.  The extent to which that is true is quantified by the \emph{condition number} of $\A$.  The condition number in $\ell_2$ may be defined several ways, resulting in the same quantity.  When generalized to arbitrary norms, those definitions differ.
We recall two natural definitions, following Demko\cite{Demko}.
\begin{defn} The \emph{non-linear condition number} of $\A : \X \to \Y$ is
\[ \tilde{\kappa}_{\X \to \Y}(\A) = \min \left\{ \frac{\|\A\|_{\X \to \Y}\|x\|_\X}{\|\A x\|_{\Y}} : \A x \neq 0 \right\}\]
The \emph{linear condition number} is defined by
\[ \kappa_{\X \to \Y}(\A) = \min \left\{\|\A\|_{\X \to \Y}\|\mathbf{G}\|_{\Y \to \X} : \mathbf{G} \in \Lin(\Y,\X) : \A\mathbf{G}\A = \A \right\} \]
\end{defn}

Of course, $\tilde{\kappa} \leq \kappa$.
Having defined the condition number, we may now state how composition affects the approximation parameters.
\begin{theorem}[Composition]\label{comp} Let $F_i$ be an $(\alpha_i,\beta_i/\tilde{\kappa})$-algorithm for $\A: \X \to \Y$, where $\A$ has non-linear condition number $\tilde{\kappa}$  Then, the composition $F_2 \circ F_1$ is an $(\alpha_1 + \alpha_2 \beta_1, \beta_1\beta_2/\tilde{\kappa}$-algorithm for the same problem.
\end{theorem}

Before proving lemma \ref{comp}, we state two useful corollaries.  The first concerns the result of recursively composing a $(\alpha,\beta/\tilde{\kappa})$-algorithm with itself.
\begin{corollary} Let $F$ be a $(\alpha,\beta/\tilde{\kappa})$-algorithm for $\beta < 1$.  Let $F^t$ be the sequence formed by iterated composition, with $F^1 = F$, $F^{t+1} = F^t \circ F$.  Then, $F^t$ is a $(\alpha/(1-\beta),\beta^t/\tilde{\kappa})$-algorithm.
\end{corollary}
\begin{proof} By induction on $t$.  To start, an  $(\alpha,\beta/\tilde{\kappa})$-algorithm is trivially a $(\alpha/(1-\beta),\beta/\tilde{\kappa})$-algorithm.  Assuming the claim holds for $F^t$, lemma \ref{comp} implies $F^{t+1}$ is a $(\alpha +\alpha \beta/(1-\beta),\beta^{t+1}/\tilde{\kappa})$-algorithm.  \end{proof}

We observe that to obtain $(1+\O(\epsilon),\delta)$-solution, only the first solver in the chain need be very accurate.
\begin{corollary} Let $F$ be a $(1+\epsilon,\epsilon/2\tilde{\kappa}$-solver and $G$ be a $(2,1/2\tilde{\kappa})$-solver.  Then, $G^t \circ F$ is a $(1+5\epsilon, \epsilon 2^{-t-1}/\tilde{\kappa}$-solver.
\end{corollary}

For some problems, there is a very simple $(M,0)$-solver known.  A final composition with that solver serves to elimate the error.

\begin{corollary} Let $F$ be a $(1+\epsilon, \epsilon\delta/\tilde{\kappa})$-solver and $G$ be a $(M,0)$-solver.  Then $G \circ F$ is a $(1+\epsilon(1+\delta M), 0)$-solver.
\end{corollary}

\subsection{Proof of lemma \ref{comp}}
Since $F_1$ is an $(\alpha_1,\beta_1/\tilde{\kappa})$-algorithm, we have
\begin{eqnarray*}
\|x\| & \leq & \alpha_1 \|\x_\opt\| \\
\frac{\|\tilde{\b}\|}{\|\A\|} & \leq & \frac{\beta_1}{\tilde{\kappa}} \|\x_\opt\| 
\end{eqnarray*}
By the definition of $\tilde{\kappa}$,
\[ \|\tilde{\x}_\opt\| \leq \tilde{\kappa} \frac{\|\tilde{\b}\|}{\|\A\|} \leq \beta_1 \|\x_{\opt}\| \]
Since $F_2$ is an $(\alpha_2,\beta_2/\tilde{\kappa})$-algorithm,
\begin{eqnarray*}
\|\tilde{x}\| & \leq & \alpha_2 \|\tilde{\x}_\opt\| \leq \alpha_2 \beta_1 \|\x\|_\opt \\
\frac{\|\A \tilde{x} - \tilde{\b}\|}{\|\A\|} & \leq & \frac{\beta_2}{\tilde{\kappa}} \|\tilde{\x}_\opt\| \leq \frac{\beta_1\beta_2}{\tilde{\kappa}}
\end{eqnarray*}
The conclusion follows from $\|x+\tilde{x}\| \leq \|x\|+\|\tilde{x}\|$.

\section{Solvers for $\ell_p$} \label{lpsolve}

The recursive composition technique takes existing $(\alpha,\beta)$-approximation algorithms and yields new algorithms with different approximation parameters.  In this section, we discuss the parameters of existing well-known base solvers.

Let $\A: \R^m \to \R^n$ be linear and fixed throughout this section.
To concisely describe results and ease comparison, for a norm $\|\cdot\|_\X$ on $\R^m$ and a norm $\|\cdot\|_\Y$ in $\R^n$, we write $(\alpha,\beta)_{\X \to \Y}$ to denote an $(\alpha,\beta)$ solution with respect to $\A: \X \to \Y$.

The classical $\ell_2$ solvers provide a useful goalpost for comparison.
\begin{theorem}[\cite{Boyd}]
\item The \emph{steepest-descent} algorithm produces a $(1,\delta)_{2 \to 2}$-solution after $\O(\kappa_{2 \to 2}(\A)^2 \log(1/\delta))$ simple iterations
\item The \emph{conjugate gradient} algorithm produces a $(1,\delta)_{2 \to 2}$-solution after $\O(\kappa_{2 \to 2}(\A) \log(1/\delta))$ simple iterations.
\end{theorem}

There are many existing algorithms for $p$-norm minimization, and composing them yields algorithms with exponentially small error.  For this initial manuscript, we simply state the $\ell_1$ version required for min-cost flow.  The general cases  consist of describing the many existing algorithms\cite{NesterovSmooth} in terms of $(\alpha,\beta)$-solvers.

\begin{theorem}\label{l1} There is a $(1+\epsilon,\delta)_{1 \to 1}$-solver with simple iterations totalling
\[ \O\left(\tilde{\kappa}_{1 \to 1}(\A)^2 \log(m) \left(\epsilon^{-2}) + \log(\delta^{-1})\right)\right) \]
\end{theorem}

\begin{theorem}\label{linf} There is a $(1+\epsilon,\delta)_{\infty \to \infty}$-solver with simple iterations totalling
\[ \O\left(\tilde{\kappa}_{\infty \to \infty}(\A)^2 \log(n) \left(\epsilon^{-2} + \log(\delta^{-1})\right)\right) \]
\end{theorem}

The preceding theorems follow by combining the \emph{multiplicative weights algorithm}\cite{PST,FS} with the composition lemma.  The former algorithm applies in a more general online setting; for our applications, the algorithm yields the ``weak'' approximation algorithms that shall be composed.

A \emph{support oracle} for a compact-convex set $\C$ takes input $y$ and returns some $x \in \C$ maximizing $x \cdot \C$.
Let $\simp_n$ be the unit simplex in $\R^n$ (i.e., the convex hull of the standard basis).

The multiplicative weights method finds approximate solutions to the saddle point problem,
\[ \max_{w \in \simp_n} \min_{z \in \C} w \cdot z \]
In particular, it obtains an additive $\epsilon$-approximation in $\O(\rho^2 \epsilon^{-2}\log n)$ iterations with each iteration taking $\O(n)$ time plus a call to a support oracle for $\C$.

Let $\B_p$ be the unit ball in $\ell_p$.  A point $w \in \B_1$ is represented as $w = w_+ - w_-$ for $(w_+,w_-) \in \simp_{2n}$.
The base solvers follow by considering the saddle-point problems,
\begin{eqnarray*}
\max_{y^* \in \b_1} \min_{x \in \b_\infty}& y^* \cdot (\A x - \mu \b) \\
\max_{x \in \b_1} \min_{y^* \in \b_\infty} & y^* \cdot (\A x - \mu \b)
\end{eqnarray*}
where $\mu$ is scaled via binary-search as needed.  The search adds a factor $\O(\log \tilde{\kappa})$ to the complexity.
This factor may be avoided by using $\ell_p$-norm for $p =\log n$ and $p = (\log m)/\epsilon$ regularization, respectively, instead.
%
\section{Generalized Preconditioning}\label{precondsec}
Let $\A : \X \to \Y$ be linear.   The minimum-norm problem for $\A$ does not intrinsically require $\Y$ to be normed; nor does the definition of a $(\alpha,0)$-solution.  Thus, an algorithm designer seeking to solve a class of minimum norm problems may freely choose how to norm $\Y$.  The best choice requires balancing two factors.  First, the norm should be sufficiently ``simple'' that $(1,\epsilon)$ solutions for polynomially-small $\epsilon$ are easy to find.  Second, the norm should be chosen so that $\tilde{\kappa}(\A)_{\X \to \Y}$ is small.  If only the latter constraint existed, the ideal choice would be the \emph{optimal $\A$ pre-image norm},\footnote{We remark this is only defined on the image space of $\A$, rather than the entire codomain.  However, we only consider $\b$ in that image space}
\[ \|\b\|_{\opt(\A)} = \min \{ \|x\|_\X : \A x = \b \} \]
By construction, $\tilde{\kappa}(\A)_{\X \to \opt(\A)} = 1$.
Of course, simply evaluating that norm is equivalent to the original problem we aim to solve.
It follows that $\Y$ should be equipped with the relatively closest norm to $\opt(\A)$ that can be efficiently minimized.

Following the common approach in $\ell_2$, when $\X = \ell_p$, we propose to choose a \emph{generalized pre-conditioner} $\P : \Y \to \ell_p$ injective and consider the norm $\|\P \b\|_p$.  Again, $\P$ should be chosen such that $\tilde{\kappa}(\P\A)_{p \to p}$ is small, and $\P,\P^*$ are easy to compute.
\section{Graph Problems}\label{graphsec}
In this section we discuss some applications of generalized preconditioning to network flow problems.
Let $\G = (V,E)$ be an undirected graph with $n$ vertices and $m$ edges.  While undirected, we assume the edges are oriented arbitrarily.  We denote by $\V^*,\E^*$ the spaces of real-valued functions $f : V \to \R$ and $f: E \to \R$ on vertices and edges, respectively. 
We denote by $\V,\E$ the corresponding dual spaces of \emph{demands} and \emph{flows}.  For $S \subseteq V$, we write $\one_S \in \V^*$ for the indicator function on $S$.

The \emph{discrete derivative operator} $\D^* : \V^*\to\E^*$ is defined by $(\D^* f)(xy) = f(y) - f(x)$ for each oriented edge $xy$. 
The constant function has zero derivative, $\D^* \one_V = 0$.
The (negative) adjoint $-\D: \E \to \V$ is the \emph{discrete divergence operator}; for a flow $\jmath \in \E$ and vertex $x \in V$, $-(\D \jmath)(x)$ is the net quantity being transported away from vertex $x$ by the flow.

A \emph{single-commodity flow problem} is specified by a \emph{demand vector} $\b \in \V$, and requires finding a flow $\jmath \in \E$ satisfying $\D \jmath = \b$, minimizing some cost function.  When the cost-function is a norm on $\E$, the problem is a minimum-norm pre-image problem.  Assuming $G$ is connected, $\b$ has a pre-image iff the total demand $\one_V \cdot \b$ is zero.  We hereafter assume $G$ to be connected and total demands equal to zero.

Having established the generalized preconditioning framework, we may now design fast algorithms for fundamental network flow problems by designing generalized preconditioners for the corresponding minimum-norm problems.  As a result, we obtain nearly-linear time algorithms for max-flow and uncapacitated min-cost flow in undirected graphs.  The max-flow algorithm was presented previously\cite{Sherman13}, and involved a combination of problem-specific definitions (e.g. \emph{congestion approximators}) and subroutines (e.g. gradient-based $\ell_\infty$ minimization).  The brief summary in subsection \ref{maxflow} shows how the same result may be immediately obtained by combining theorem \ref{l1linf} with only small part of the earlier flow-specific work\cite{Sherman13}: the actual construction of the preconditioner.

In section \ref{mincostsec}, we present a nearly-linear time algorithm for uncapacitated min-cost flows, another fundamental network optimization problem.  Once again, all that is needed is description and analysis of the preconditioner.
\subsection{Max-Flow}\label{maxflow}
A \emph{capacitated graph} associates with each edge $e$, a \emph{capacity} $\c(e) > 0$.  
Given capacities, we define the \emph{capacity} norm on $\E^*$ by $\|f\|_{\c*} = \sum_e \c(e)|f(e)|$; 
the dual \emph{congestion norm} on $\E$ is defined by $ \|\jmath \|_{\c} = \max_e \frac{|\jmath(e)|}{\c(e)}$.

The \emph{boundary capacity} of a set $S \subseteq \V$ is $\|\D^* \one_S\|_{c}$, which we abbreviate as $\c(\partial S)$.

The single-commodity flow problem with capacity norm is \emph{undirected maximum flow problem}, and is a fundamental problem in network design and optimization.
Following the work of Spielman and Teng for $\ell_2$-flows\cite{ST}, and Madry\cite{Madry10} for cut problems, the author has obtained a nearly-linear time algorithm for maximum flow\cite{Sherman13}.  We have since realized that only certain parts of the ideas introduced in that paper are actually specific to maximum-flow, and the present work has been obtained through generalizing the other parts.  The (previously presented as) flow-specific composition and optimization parts of that paper are entirely subsumed by the earlier sections of this paper.  The truly flow-specific part of that work required to complete the algorithm is the \emph{generalized preconditioner} construction, which we have previously called a \emph{congestion-approximator}.  We briefly describe the main ideas of the construction, and some simple illustrative examples.  For the full details, we refer the reader to the third section of \cite{Sherman13}.

Let $\FF \subseteq 2^\V$ be a family of subsets.  Such a family, together with a capacity norm on $\E$, induces a \emph{cut-congestion} semi-norm on $\V$ via $\|b\|_{\c(\FF)} = \max_{S \in \FF} \frac{|\one_S \cdot b|}{\c(\partial S)}$.  That is, for each set in the family,  consider the ratio of the total aggregate demand in that set to the total capacity of all edges entering that set.  If the indicator functions of sets in $\FF$ span $\V$, then it is a norm.  Note also that we may write $\|b\|_{c(\FF)} = \|\P b\|_\infty$, where $\P : \R^n \to \R^{\FF}$ is defined by $(\P b)_S = \frac{\one_S \cdot b}{\c(\partial S)}$.

If $\FF$ is taken to be the full power set $\FF = 2^\V$, the \emph{max-flow, min-cut theorem} is equivalent to the statement that the non-linear condition number $\tilde{\kappa}_{\c \to \c(\FF)}(\D^*)$ is exactly one.  If $\FF$ is taken to be the family of singleton sets, the non-linear condition number is the inverse of the \emph{combinatorial conductance} of $\G$.  That is, for high-conductance graphs, composition yields a fast algorithm requiring only a simple \emph{diagonal} preconditioner.  This is closely analagous to the situation for $\ell_2$, where diagonal preconditioning yields fast algorithms for graphs of large \emph{algebraic conductance}.

A particularly illustrative example consists of a unit-capacity $2^t \times 2^t$ 2-D grid.  With low $(\Theta(2^{-t}) = \Theta(1/\sqrt{n})$ conductance, the singleton family $\FF$ performs poorly.  A simple but dramatically better family consists of taking $\FF$ to be all power-of-two size, power-of-two aligned subgrids.  The latter family yields \emph{linear} condition number $\O(t) = \O(\log n)$.  Furthermore, computing the aggregate demand in all such sets requires only $\O(n)$ time, with each sub-grid aggregate equal to the sum of its four child aggregates.  This algorithm is analagous to an $\ell_\infty$ version of \emph{multi-grid}\cite{multi}.  Note that the indicator ``step'' functions $\one_S$ are different from the bilinear ``pyramid'' functions used by multigrid.

\begin{theorem}[\cite{Sherman13}] There is an algorithm that given a capacitated graph $(G, \c)$ with $n$ vertices and $m$ edges, takes $m^{1+o(1)}$ time and outputs a data structure of size $n^{1+o(1)}$ that efficiently represents a preconditioner $\P$ with $\kappa(\P \D)_{\c \to \infty} \leq n^{o(1)}$.  Given the data structure,
$\P$ and $\P^*$ can be applied in $n^{1+o(1)}$ time.
\end{theorem}

Let $\C : \R^m \to \E$ be the diagonal capacity matrix; $(\C x)(e) = \c(e) x(e)$.  The nearly-linear time algorithm follows by using the preconditioner $\P$ in the preceding theorem, applying theorem \ref{l1linf} to the problem of minimizing $\|x\|_\infty$ subject to $\P \D \C x = \b$.  
A final minor flow-specific part is needed to achieve zero error: the simple $(m,0)$-approximation algorithm that consists of routing all flow through a maximum-capacity spanning tree.


\section{Uncapacitated Min-Cost Flow}\label{mincostsec}
In this section, we present a nearly-linear time algorithm for the uncapacitated minimum-cost flow problem on undirected graphs.  A \emph{cost graph} associates a \emph{cost} $\l(e) > 0$ with each edge $e$ in $G$.  The \emph{cost} norm on the space of flows $\E$ is defined by $\|\jmath\|_{\l} = \sum_e \|\jmath(e)\|\l(e)$.
The \emph{undirected uncapacitated minimum-cost flow problem} is specified by the graph $G$, lengths $\ell$, and demands $\b \in \V$, and consists of finding $\jmath \in \E$ with $\D \jmath = \b$ minimizing $\|\jmath\|_{\l}$.
Note that, unlike max-flow, the \emph{single-source} uncapacitated min-cost flow problem is significantly easier than the distributed source case: the optimal solution is to route along a single-source shortest path away from the source.  The problem only becomes interesting with distributed sources and sinks.

As a new application of generalized preconditioning, we obtain a nearly-linear time algorithm for approximately solving this problem.
\begin{theorem}\label{mincost}  There is a randomized algorithm that, given a length-graph $G = (V,E,\l)$, takes $\frac{m^{1+o(1)}}{\epsilon^2}$ time and outputs a $(1+\epsilon,0)$-solution to the undirected uncapacitated min-cost flow problem.
\end{theorem}

As expected, the bulk of the algorithm lies in constructing a good preconditioner.
After introducing some useful definitions and lemmas, we spend the majority of this section constructing and analyzing a preconditioner for this problem.
\begin{theorem}\label{minpre} There is a randomized algorithm that, given a length-graph $G$, takes $\O(m\log^2 n + n^{1+o(1)})$ time and outputs a $n^{1+o(1)} \times n$ matrix $\P$ with $\kappa_{\l to \infty}(\P \D) \leq n^{o(1)}$.
Every column of $\P$ has $n^{o(1)}$ non-zero entries.
\end{theorem}

Dual to cost is the \emph{stretch} norm on $\E^*$ is $\|f\|_{\l*} = \max_e |f(e)|/\l(e)$.  We say a function $\phi \in \V$ is $L$-\emph{Lipschitz} iff $\|\D^* \phi\|_{\l*} \leq L$; that is, $|\phi(x) - \phi(y)| \leq L \l(xy)$ for all edges $xy$.  The dual problem is to maximize $\phi \cdot \b$ over $1$-Lipschitz $\phi \in \V$.

A length function $\l$ induces an intrinsic metric $\d : V \times V \to \R$, with distances determined by shortest paths.  The preceding definition of Lipschitz is equivalent to $|\phi(x) - \phi(y)| \leq L \d(x,y)$ for all $x,y \in V$.  That is, the Lipschitz constant of $\phi$ depends only on the metric, and is otherwise independent of the particular graph or edge lengths inducing that metric.
As the dual problem depends only on the induced metric, so must the \emph{value} of the min-cost flow.  Therefore, we may unambiguously write
\[ \|b\|_{\opt(\d)} = \max \{\phi \cdot b : \phi \textrm{ is 1-Lipschitz w.r.t. }\d\} \]

The monotonicity of cost with distance follows immediately.
\begin{lemma} If $\d(x,y) \leq \tilde{\d}(x,y)$ for all $x,y \in \V$, then for all $b \in \V_{\perp \one}$,
 \[\|b\|_{\opt(\d)} \leq \|b\|_{\opt(\tilde{\d})} \]
 \end{lemma}

We recall that in some cases, the identity operator is a good preconditioner; for max-flow, such is the case in constant-degree expanders.  For min-cost flow, the analagous case is a metric where distances between any pair of points differ relatively by small factors.
\begin{defn} Let  $U \subseteq V$.  We say $U$ is $r$-\emph{separated} if $\d(x,y) \geq r$ for all $x,y \in U$ with $x \neq y$.  We say $U$ is $R$-\emph{bounded} if $\d(x,y) \leq R$ for all $x,y \in U$.
\end{defn}

\begin{lemma}\label{sepdiam} Let $U \subset V$ be $R$-bounded and $r$-separated with respect to $\d$.  Let $b \in \V_{\perp \one}$ be demands supported on $U$.  Then,
\[ \frac{r}{2} \|b\|_1 \leq \|b\|_{\opt(\d)} \leq \frac{R}{2}\|b\|_1 \] 
\end{lemma} \begin{proof} Note $\frac{r}{2}\|b\|_1$ is exactly the min-cost of routing $b$ in the metric where all distances are exactly $r$ (consider the star graph with edge lengths $r/2$, with a leaf for each $x \in U$).  The inequalities follow by monotonicity.
\end{proof}

To complete the proof of theorem \ref{mincost}, assuming theorem \ref{minpre}, we also need a simple approximation algorithm that yields a poor approximation, but with zero error.  This is easily achieved by the algorithm that routes all flow through a minimum cost (i.e., total length) spanning tree $T$.
\begin{lemma}\label{mintree} Let $T$ be a minimum cost spanning tree with respect to $\l$ on $\G$.  Consider the algorithm that, given $b \in \V_{\perp \one}$, outputs the flow $\jmath \in \E$ routing $b$ using only the edges in $T$.  Then,
\[ \|\jmath\| \leq n \|b\|_{\opt(\d)} \]
\end{lemma} \begin{proof}  The flow on $T$ is unique, and therefore, optimal for the metric $\tilde{\d}$ induced by paths restricted to lie in $T$.  As $T$ is a minimum cost spanning tree, $\tilde{\d}(x,y) \leq n \d(x,y)$.
\end{proof}

We now prove theorem \ref{mincost}.  Let $\L : \R^m \to \E$ be the diagonal length matrix $(\L x)(e) = \ell(e) x(e)$, and let $\P$ be the preconditioner from theorem \ref{minpre}.   Computing $\P$ takes $\O(m\log^2 n) + n^{1+o(1)}$ time.  The minimum-cost flow problem is equivalent to the problem of minimizing $\|x\|_1$ subject to $\P\D\L^{-1} x = \P \b$.  As $\P$ has $n^{1+o(1)}$ non-zero entries and $\kappa_{1 \to 1}(\P\D\L^{-1}) \leq n^{o(1)}$, the total running time is $\frac{m^{1+o(1)}}{\epsilon^2}$.

\section{Preconditioning Min-Cost Flow}
A basic concept in metric spaces is that of an \emph{embedding}: $\psi: V \to \tilde{V}$, with respective metrics $\d,\tilde{\d}$.  An embedding has \emph{distortion} $L$ if for some $\mu > 0$,
 \[ 1 \leq \mu \frac{\tilde{\d}(\psi(x),\psi(y))}{\d(x,y)} \leq L \]

Metric embeddings are quite useful for preconditioning; if $\d$ embeds into $\tilde{\d}$ with distortion $L$, then the monotonicity lemma implies the relative costs of flow problems differ by a factor $L$.  It follows that if we can construct a preconditioner $\P$ for the min-cost flow problem on $\tilde{\d}$, the same $\P$ may be used to precondition min-cost flow on $\d$, with condition number at most $L$-factor worse.  

Our preconditioner uses embeddings in two ways.  The preconditioner itself is based on a certain hierarchical routing scheme.  Hierarchical routing schemes based on embeddings into tree-metrics have been studied extensively(see e.g. \cite{FRT}), and immediately yield preconditioners with polylogarithmic condition number.  However, we are unable to efficiently implement those schemes in nearly-linear time.  Instead, we apply Bourgain's $\O(\log n)$-distortion embedding\cite{Bourgain85} into $\ell_1$.  That is, by paying an $\O(\log n)$ factor in condition number, we may completely restrict our attention to approximating min-cost flow in $\ell_1$ space.   

\begin{theorem}[Bourgain's Embedding\cite{Bourgain85}]\label{BGlemma}
Any $n$-point metric space embeds into $\ell_p$ space with distortion $\O(\log n)$.
If $\d$ is the intrinsic metric of a length-graph $G = (V,E)$ with $m$ edges, there is a randomized algorithm that takes $\O(m\log^2 n)$ time and w.h.p. outputs such an embedding, with dimension $\O(\log^2 n)$.
\end{theorem}


The complexity of constructing and evaluating our preconditioner is \emph{exponential} in the embedding's dimension.  Therefore, we reduce the dimension to $\Theta(\sqrt{\log n})$, incuring another distortion factor of $\exp(\O(\sqrt{\log n}))$.

\begin{theorem}[Johnson-Lindenstrauss Lemma\cite{JL,Dasgupta}]  \label{JLlemma}Any $n$ points in Euclidean space embed into $k$-dimensional Euclidean space with distortion $n^{\O(1/k)}$, for $k \leq \Omega(\log n)$.  In particular, $k = \O(\log n)$ yields $\O(1)$ distortion, and $k = \O(\sqrt{\log n})$ yields $\exp(\O(\sqrt{\log n}))$ distortion.

If the original points are in $d > k$ dimensions, there is a randomized algorithm that w.h.p outputs such an embedding in $\O(nd)$ time.
\end{theorem}

Our reduction consists of first embedding the vertices into $\ell_2$ via theorem \ref{BGlemma}, reducing the dimension to $k = \O(\sqrt{\log n})$ using lemma \ref{JLlemma}, and then using the identity embedding of $\ell_2$ into $\ell_1$, for a total distortion of $\exp(O(\sqrt{\log n}))$.

To estimate routing costs in $\ell_1$, we propose an extremely simple hierarchical routing algorithm in $\ell_1$, and show it achieves a certain competetive ratio.  The routing scheme applies to the exponentially large graph corresponding to a lattice discretization of the $k$-dimensional cube in $\ell_1$; after describing such a scheme, we show that if the demands in that graph are supported on $n$ vertices, the distributed routing algorithm finds no demand to route in most of the graph, and the same routing algorithm may be carried out in a $k-d$-tree instead of the entire lattice.

On the lattice, the routing algorithm consists of sequentially reducing the support set of the demands from a lattice with spacing $2^{-{t-1}}$ to a lattice with spacing $2^{-t}$, until all demand is supported on the corners of a cube containing the original demand points.  On each level, the demand on a vertex not appearing in the coarser level is eliminated by pushing its demand to a distribution over the aligned points nearby.  

\subsection{Lattice Algorithm}
For $t \in \Z$, let $V_t = (2^{-t}\Z)^k$ be the lattice with spacing $2^{-t}$.  We design a routing scheme and preconditioner for min-cost flow in $\ell_1$, with initial input demands $b_T$ supported on a bounded subset of $V_T$.  For $t = T,T-1,\ldots$, the routing scheme recursively reduces a min-cost flow problem with demands supported on $V_t$ to a min-cost flow problem with demands supported on the sparser lattice $V_{t-1}$.  Given demands $b_t$ supported on $V_t$, the scheme consists of each vertex $x \in V_t$ pulling its demand $b_t(x)$ uniformly from the closest points to $x$ in $V_{t-1}$ along any shortest path, the lengths of which are at most $k 2^{-t}$.
For example, a point $x = (x_1,\ldots,x_k) \in V_1$ with $j \leq k$ non-integer coordinates and demand $b_1(x)$ pulls $b_1(x)2^{-j}$ units of flow from each of the $2^j$ points in $V_0$ corresponding to rounding those $j$ coordinates in any way.  By construction, all points in $V_{t} \setminus V_{t-1}$ have their demands met exactly, and we are left with a reduced problem residual demands $b_{t-1}$ supported on $V_{t-1}$.  As we show, eventually the demands are supported entirely on the $2^k$ corners of a hypercube forming a bounding box of the original support set, at which point the the simple scheme of distributing the remaining demand uniformly to all corners is a $k$-factor approximation.

We do not actually carry out the routing; rather, we state a crude upper-bound on its cost, and then show that bound itself exceeds the true min-cost by a factor of some $\kappa$.   For the preconditioner, we shall only require the sequence of reduced demands $b_t,b_{t-1},\ldots$, and not the flow actually routing them.  

\begin{theorem}\label{prenorm}
\[ \|b_t\|_{\opt} \leq \sum_{s = 0}^t k2^{-s} \|b_s\|_1 \leq 2k(t+1) \|b_t\|_{\opt} \]
\end{theorem}
Assuming theorem \ref{prenorm} holds, we take our preconditioner to be a matrix $\P$ such that,
\[ \|\P b_T\|_1 = \sum_{t = 0}^T k 2^{-t}\|b_t\|_1 \]

The proof of theorem \ref{prenorm} uses two essential properties of the routing scheme described.
The first is that the cost incurred in the reduction from $b_t$ to $b_{t-1}$ is not much larger than the min-cost of routing $b_t$.  The second is the reduction does not increase costs.  That is, the true min-cost for routing the reduced problem $b_{t-1}$ is at most that of routing the original demands $b_t$.  
\begin{lemma}\label{indstep}  Let $b_{t-1}$ be the reduced demands produced when the routing scheme is given $b_t$.  Then,
\begin{enumerate}
\item The cost incurred by that level of the scheme is at most $k2^{-t} \|b_t\|_1 \leq 2k \|b_t\|_{\opt}$
\item The reduction is non-increasing with respect to min-cost:
$ \|b_{t-1}\|_\opt \leq \|b_t\|_{\opt} $
\end{enumerate}
\end{lemma}

Assuming lemma \ref{indstep}, we now prove \ref{prenorm}.  We argue the sum is an upper bound on the total cost incurred by the routing scheme, when applied to $b_t$.  The left inequality follows immediately, as the true min-cost is no larger.

The total cost of the scheme is at most the cost incurred on each level, plus the cost of the final routing of $b_0$.  Lemma \ref{indstep}(a) accounts for all terms above $s = 0$.  For the final routing, the demands $b_0$ are supported in $\{0,1\}^k$, so the cost of routing all demand to a random corner is at most $\frac{1}{2}k \|b_0\|_1$. 

For the right inequality, we observe
\[ \sum_{s = 0}^t k 2^{-t} \|b_t\|_1 \leq \sum_{s = 0}^t 2k \|b_s\|_{\opt} \]
The inequality follows termwise, using lemma \ref{indstep}(a) for $s \geq 1$; the $s=0$ case follows from lemma \ref{sepdiam} because $V_0$ is $1$-separated.

\subsection{Proof of Lemma \ref{indstep}}
For the first part, each vertex $x\in V_t$ routes its demand $b_t(x)$ to the closest points in $V_{t-1}$ any shortest path.  As such paths have length at most $k2^{-t}$, each point $x$ contributes at most $k2^{-t}|b_t(x)|$.

For the second part, by scale-invariance it suffices to consider $t = 1$.  Moreover, it suffices to consider the case where $b_1$ consists of a unit demand from between two points $x,y \in V_1$ with $\|x - y\|_1 = \frac{1}{2}$.  To see this, we observe any flow routing arbitrary demands $b_1$ consists of a sum of such single-edge flows.  As the reduction is linear, if the cost of each single-edge demand-pair is not increased, then the cost their sum is not increased.
By symmetry, it suffices to consider $x,y$ with $x = (0,z)$ and $y = (\frac{1}{2},z)$.  Then, the reduction distributes $x$'s demand to a distribution over $(0,z')$ where $z' \sim Z'$; $y$'s demand is split over $(0,z'$) and $(1,z')$ where $z' \sim Z'$.  Therefore, half of the demand at $(0,z')$ is cancelled, leaving the residual problem of routing $1/2$ unit from $(0,z')$ to $(1,z')$.  That problem has cost $1/2$, by routing each fraction directly from $(0,z')$ to $(1,z')$.

\ifFOCS
\else
\bibliographystyle{plain}
\fi
\bibliography{multi}

\begin{thebibliography}{10}

\bibitem{AGOT}
Ravindra~K Ahuja, Andrew~V Goldberg, James~B Orlin, and Robert~E Tarjan.
\newblock Finding minimum-cost flows by double scaling.
\newblock {\em Mathematical programming}, 53(1-3):243--266, 1992.

\bibitem{BH}
Josh Barnes and Piet Hut.
\newblock A hierarchical o (n log n) force-calculation algorithm.
\newblock {\em nature}, 324(6096):446--449, 1986.

\bibitem{Bourgain85}
J.~Bourgain.
\newblock On lipschitz embedding of finite metric spaces in hilbert space.
\newblock {\em Israel Journal of Mathematics}, 52(1):46--52, 1985.

\bibitem{Boyd}
Stephen Boyd and Lieven Vandenberghe.
\newblock {\em Convex Optimization}.
\newblock Cambridge University Press, New York, NY, USA, 2004.

\bibitem{DaitchSpielman}
Samuel~I Daitch and Daniel~A Spielman.
\newblock Faster approximate lossy generalized flow via interior point
  algorithms.
\newblock In {\em Proceedings of the fortieth annual ACM symposium on Theory of
  computing}, pages 451--460. ACM, 2008.

\bibitem{Dasgupta}
Sanjoy Dasgupta and Anupam Gupta.
\newblock An elementary proof of a theorem of johnson and lindenstrauss.
\newblock {\em Random Structures \& Algorithms}, 22(1):60--65, 2003.

\bibitem{Demko}
Stephen Demko.
\newblock Condition numbers of rectangular systems and bounds for generalized
  inverses.
\newblock {\em Linear Algebra and its Applications}, 78:199--206, 1986.

\bibitem{FRT}
Jittat Fakcharoenphol, Satish Rao, and Kunal Talwar.
\newblock A tight bound on approximating arbitrary metrics by tree metrics.
\newblock In {\em Proceedings of the thirty-fifth annual ACM symposium on
  Theory of computing}, pages 448--455. ACM, 2003.

\bibitem{FS}
Yoav Freund and Robert~E Schapire.
\newblock Adaptive game playing using multiplicative weights.
\newblock {\em Games and Economic Behavior}, 29(1):79--103, 1999.

\bibitem{JL}
William~B Johnson and Joram Lindenstrauss.
\newblock Extensions of lipschitz mappings into a hilbert space.
\newblock {\em Contemporary mathematics}, 26(189-206):1, 1984.

\bibitem{KLOS}
Jonathan~A Kelner, Yin~Tat Lee, Lorenzo Orecchia, and Aaron Sidford.
\newblock An almost-linear-time algorithm for approximate max flow in
  undirected graphs, and its multicommodity generalizations.
\newblock In {\em Proceedings of the Twenty-Fifth Annual ACM-SIAM Symposium on
  Discrete Algorithms}, pages 217--226. Society for Industrial and Applied
  Mathematics, 2014.

\bibitem{LeeSidford}
Yin~Tat Lee and Aaron Sidford.
\newblock Path finding methods for linear programming: Solving linear programs
  in o (vrank) iterations and faster algorithms for maximum flow.
\newblock In {\em Foundations of Computer Science (FOCS), 2014 IEEE 55th Annual
  Symposium on}, pages 424--433. IEEE, 2014.

\bibitem{Madry10}
Aleksander Madry.
\newblock Fast approximation algorithms for cut-based problems in undirected
  graphs.
\newblock In {\em Foundations of Computer Science (FOCS), 2010 51st Annual IEEE
  Symposium on}, pages 245--254. IEEE, 2010.

\bibitem{NesterovSmooth}
Yu~Nesterov.
\newblock Smooth minimization of non-smooth functions.
\newblock {\em Math. Program.}, 103(1):127--152, May 2005.

\bibitem{PST}
Serge~A. Plotkin, David~B. Shmoys, and Eva Tardos.
\newblock Fast approximation algorithms for fractional packing and covering
  problems.
\newblock {\em Mathematics of Operations Research}, 20:257--301, 1995.

\bibitem{Sherman13}
Jonah Sherman.
\newblock Nearly maximum flows in nearly linear time.
\newblock In {\em Foundations of Computer Science (FOCS), 2013 IEEE 54th Annual
  Symposium on}, pages 263--269. IEEE, 2013.

\bibitem{ST}
Daniel~A. Spielman and Shang-Hua Teng.
\newblock Nearly-linear time algorithms for preconditioning and solving
  symmetric, diagonally dominant linear systems.
\newblock {\em CoRR}, abs/cs/0607105, 2006.

\end{thebibliography}

\appendix

\end{document}